\def\llncs{0}       % 1 = using llncs
\def\showoverflow{0}% 1 = show overflows when llncs=1
\def\setmargin{1}   % 0 = do not set; 1 = set
\def\papertype{1}   % 0 = letter; 1 = a4;
\def\savespace{0}   % 0 = do not save space; 1 = save space
\def\myfontsize{1}  % 0 = 10pt; 1 = 11pt; 0 for llncs
\def\anonymous{0}   % 0 = non-anonymous; 1 = anonymous
\def\showkeywords{1} % 0 = turn keywords off when llncs=0
\def\showdate{0} % 1 = show date when llncs=anonymous=0
\def\instplace{1} % 0 = institutes as footnotes; 1= inline
\def\instcompact{1} % 1= merge the inline institute information of co-located authors
\def\addtoc{0}      % 0 = exclude TOC; 1 = include TOC
\def\fullversion{1} % 0 = short version; 1 = full version
\def\authnotes{0}   % 0 = authnotes off; 1 = authnotes on

\def \cmid{\mspace{1.5mu}\vert\mspace{1.5mu}}

\ifnum\llncs=1
  \ifnum\papertype=0
    \ifnum\myfontsize=0
      \ifnum\showoverflow=1
        \documentclass[twoside,envcountsame,runningheads,draft]{llncs}
      \else
        \documentclass[twoside,envcountsame,runningheads]{llncs}
      \fi
    \else
      \ifnum\showoverflow=1
        \documentclass[11pt,twoside,envcountsame,runningheads,draft]{llncs}
      \else
        \documentclass[11pt,twoside,envcountsame,runningheads]{llncs}
      \fi
      \usepackage[letterpaper,hmargin=1in,vmargin=1.25in]{geometry}
    \fi
  \else
    \ifnum\myfontsize=0
      \ifnum\showoverflow=1
        \documentclass[a4paper,twoside,envcountsame,runningheads,draft]{llncs}
      \else
        \documentclass[a4paper,twoside,envcountsame,runningheads]{llncs}
      \fi
    \else
      \ifnum\showoverflow=1
        \documentclass[a4paper,11pt,twoside,envcountsame,runningheads,draft]{llncs}
      \else
        \documentclass[a4paper,11pt,twoside,envcountsame,runningheads]{llncs}
      \fi
    \fi
  \fi
\else
  \ifnum\papertype=0
    \ifnum\myfontsize=0
      \documentclass[twoside]{article}
    \else
      \documentclass[11pt,twoside]{article}
    \fi
  \else
    \ifnum\myfontsize=0
      \documentclass[a4paper,twoside]{article}
    \else
      \documentclass[11pt,a4paper,twoside]{article}
    \fi
  \fi
\fi

\usepackage{url}
\usepackage{latexsym}
\usepackage{amscd,amsmath}
\usepackage{amsfonts}
\usepackage{color}
\usepackage{float}
\usepackage{longtable}
\usepackage{graphicx}
\usepackage{xspace}

\ifnum\savespace=1
\usepackage{times}
\fi

\ifnum\llncs=1
  \usepackage{graphics}
\else
  \ifnum\anonymous=0
    \usepackage[
    pdftitle={},
    pdfauthor={},
    pdfpagelabels=true, 
    linktocpage=true,
    pagebackref
    ]{hyperref}
  \else
    \usepackage[
    pdftitle={},
    pdfauthor={No Authors},
    pdfpagelabels=true, 
    linktocpage=true,
    ]{hyperref}
  \fi
\fi

\ifnum\setmargin=1
  \ifnum\papertype=1
  \else
  \fi
\fi

\ifnum\llncs=0\ifnum\fullversion=1\ifnum\anonymous=0
  
\fi\fi\fi

\ifnum\llncs=0

\makeatletter
\def\appearsin#1{\gdef\@appearsin{#1}}
\def\maketitle{\par
 \begingroup
 \def\thefootnote{\arabic{footnote}}
 \def\@makefnmark{\hbox
 to 0pt{$^{\@thefnmark}$\hss}}
 \if@twocolumn
 \twocolumn[\@maketitle]
 \else \newpage
 \global\@topnum\z@ \@maketitle \fi\thispagestyle{plain}\@thanks
 \endgroup
 \setcounter{footnote}{0}
 \let\maketitle\relax
 \let\@maketitle\relax
 \gdef\@thanks{}\gdef\@author{}\gdef\@title{}\gdef\@appearsin{}
          \let\thanks\relax}
\def\@maketitle{\newpage
 \noindent \@appearsin
 \vskip 0.5in \begin{center}
 {\LARGE \@title \par} \vskip 1.5em {\large \lineskip .5em
\begin{tabular}[t]{c}\@author
 \end{tabular}\par}
 \vskip 1em {\normalsize \@date} \end{center}
 \par
 \vskip 1.5em}
\def\abstract{\if@twocolumn
\section*{Abstract}
\else \small
\begin{center}
{\bf Abstract\vspace{-.5em}\vspace{0pt}}
\end{center}
\quotation
\fi}
\def\endabstract{\if@twocolumn\else\endquotation\fi}
\mark{{}{}}
\fi

\DeclareMathAlphabet{\mathsl}{OT1}{cmr}{m}{sl}
\DeclareMathAlphabet{\mathsc}{OT1}{cmr}{m}{sc}
\DeclareMathAlphabet{\mathslbf}{OT1}{cmr}{bx}{sl}
\DeclareFontFamily{OT1}{pzc}{}
\DeclareFontShape{OT1}{pzc}{m}{it}%
             {<-> s * [1.150] pzcmi7t}{}
\DeclareMathAlphabet{\mathscript}{OT1}{pzc}{m}{it}

\ifnum\llncs=0

\newtheorem{thm}{Theorem}[section]
\newtheorem{lem}[thm]{Lemma}
\newtheorem{cor}[thm]{Corollary}
\newtheorem{propo}[thm]{Proposition}
\newtheorem{clm}[thm]{Claim}
\newtheorem{defn}[thm]{Definition}
\newtheorem{assm}[thm]{Assumption}
\newtheorem{rem}[thm]{Remark}
\newtheorem{obs}[thm]{Observation}
\newtheorem{egs}[thm]{Example}
\newtheorem{expr}{Experiment}
\newtheorem{fct}[thm]{Fact}
\newtheorem{cons}[thm]{Construction}
\newtheorem{nte}[thm]{Note}

\newenvironment{theorem}{\begin{thm}}{\end{thm}}
\newenvironment{lemma}{\begin{lem}}{\end{lem}}
\newenvironment{corollary}{\begin{cor}}{\end{cor}}
\newenvironment{proposition}{\begin{propo}}{\end{propo}}
\newenvironment{definition}{\begin{defn}}{\end{defn}}
\newenvironment{assumption}{\begin{assm}\begin{em}}{\end{em}\end{assm}}
\newenvironment{claim}{\begin{clm}\begin{rm}}{\end{rm}\end{clm}}
\newenvironment{remark}{\begin{rem}\begin{em}}{\end{em}\end{rem}}

\newenvironment{fact}{\begin{fct}\begin{em}}{\end{em}\end{fct}}
\newenvironment{construction}{\begin{cons}\begin{rm}}{\end{rm}\end{cons}}

\else

\ifnum\fullversion=1
\newtheorem{thm}{Theorem}
\newtheorem{lem}[thm]{Lemma}
\newtheorem{cor}[thm]{Corollary}
\newtheorem{propo}[thm]{Proposition}
\newtheorem{clm}[thm]{Claim}
\newtheorem{defn}[thm]{Definition}
\newtheorem{assm}[thm]{Assumption}
\newtheorem{rem}[thm]{Remark}
\newtheorem{fct}[thm]{Fact}

\newtheorem{cons}[thm]{Construction}

\renewenvironment{theorem}{\begin{thm}}{\end{thm}}
\renewenvironment{lemma}{\begin{lem}}{\end{lem}}

\fi

\fi

\newenvironment{sfigure}{\begin{figure}[t]\begin{small}}{\end{small}\end{figure}}

\newcommand{\secref}[1]{Section~\ref{#1}}

\newcommand{\thref}[1]{Theorem~\ref{#1}}

\ifnum\llncs=1
  \newcommand{\figref}[1]{Fig.~\ref{#1}}
\else
  \newcommand{\figref}[1]{Figure~\ref{#1}}
\fi

\renewcommand{\eqref}[1]{\mbox{Equation~(\ref{#1})}}

\newlength{\saveparindent}
\newlength{\saveparskip}
\ifnum\llncs=0
\newcount\proofqeded
\newcount\proofended
\def\qed{{\hspace{1pt}\rule[-1pt]{3pt}{9pt}}
\end{rm}\addtolength{\parskip}{-0pt}
\setlength{\parindent}{\saveparindent}
\global\advance\proofqeded by 1 }
\def\qedenv{
\end{rm}\addtolength{\parskip}{-0pt}
\setlength{\parindent}{\saveparindent}
\global\advance\proofqeded by 1 }
\newenvironment{proof}%
 {\proofstart}%
 {\ifnum\proofqeded=\proofended~\qed\fi \global\advance\proofended by 1
  \medskip}
 {\proofenvstart}%
 {\ifnum\proofqeded=\proofended\qedenv\fi \global\advance\proofended by 1
  \medskip}
\makeatletter
\def\proofstart{\@ifnextchar[{\@oprf}{\@nprf}}
\def\proofenvstart{\@ifnextchar[{\@osprf}{\@nsprf}}
\def\@oprf[#1]{\begin{rm}\protect\vspace{6pt}\noindent{\bf Proof of #1:\ }%
\addtolength{\parskip}{5pt}\setlength{\parindent}{0pt}}
\def\@osprf[#1]{\begin{rm}\protect\vspace{6pt}\noindent
\addtolength{\parskip}{5pt}\setlength{\parindent}{0pt}}
\def\@nprf{\begin{rm}\protect\vspace{6pt}\noindent{\bf Proof:\ }%
\addtolength{\parskip}{5pt}\setlength{\parindent}{0pt}}
\def\@nsprf{\begin{rm}\protect\vspace{6pt}\noindent%
\addtolength{\parskip}{5pt}\setlength{\parindent}{0pt}}
\fi

\newcounter{ctr}

\newcounter{ectr}

\newlength{\savejot}
\ifnum\llncs=0

\newenvironment{newmath}{\begin{displaymath}%
\setlength{\abovedisplayskip}{6pt}%
\setlength{\belowdisplayskip}{6pt}%
\setlength{\abovedisplayshortskip}{8pt}%
\setlength{\belowdisplayshortskip}{8pt} }{\end{displaymath}}

\newenvironment{neweqnarrays}{\begin{eqnarray*}%
\setlength{\abovedisplayskip}{-4pt}%
\setlength{\belowdisplayskip}{-4pt}%
\setlength{\abovedisplayshortskip}{-4pt}%
\setlength{\belowdisplayshortskip}{-4pt}%
\setlength{\jot}{-0.4in} }{\end{eqnarray*}}

\newenvironment{newequation}{\begin{equation}%
\setlength{\abovedisplayskip}{6pt}%
\setlength{\belowdisplayskip}{6pt}%
\setlength{\abovedisplayshortskip}{8pt}%
\setlength{\belowdisplayshortskip}{8pt} }{\end{equation}}

\else

\fi

\newcommand{\authnote}[2]{
\ifnum\authnotes=1 
  \begin{center}
    \fbox{\begin{minipage}{5.7in}
      \textbf{#1 says:} #2
    \end{minipage}}
  \end{center} 
\fi
}

\ifnum\llncs=1

\else

\fi

\newcommand{\calF}{{\cal F}}

\def\bits{\{0,1\}}

\def\getsr{\stackrel{{\scriptscriptstyle\$}}{\leftarrow}}

\def\next{\,;\,}

\newcommand{\secpar}{\lambda}

\newcommand{\adv}{\ensuremath{\mathcal{A}}\xspace}
\newcommand{\david}{\ensuremath{\mathcal{D}}\xspace}
\newcommand{\goliath}{\ensuremath{\mathcal{G}}\xspace}
\newcommand{\token}{\ensuremath{\mathcal{T}}\xspace}
\newcommand{\inp}{\ensuremath{x}\xspace}
\newcommand{\out}{\ensuremath{y}\xspace}
\newcommand{\sid}{\textsf{sid}}
\newcommand{\Fstate}{\ensuremath{F_{state}}\xspace}
\newcommand{\Fabort}{\ensuremath{F_{abort}}\xspace}

\newcommand{\Tstate}{\ensuremath{T_{state}}\xspace}
\newcommand{\Swait}{\textsf{wait}\xspace}
\newcommand{\Ssent}{\textsf{sent}\xspace}
\newcommand{\Sready}{\textsf{ready}\xspace}
\newcommand{\Sdead}{\textsf{dead}\xspace}
\newcommand{\Scom}{\textsf{committed}\xspace}
\newcommand{\SMcom}{\textsf{message committed}\xspace}
\newcommand{\SMopen}{\textsf{message opened}\xspace}
\newcommand{\SIcom}{\textsf{inputs committed}\xspace}

\newcommand{\Mcreate}[1]{(\textsf{create}, \textsf{sid}, \ensuremath{\goliath, \david, #1})\xspace}
\newcommand{\Moverwrite}[1]{(\textsf{overwrite}, \textsf{sid}, \ensuremath{\goliath, \david, #1})\xspace}
\newcommand{\Mcreated}{(\textsf{created}, \textsf{sid}, \ensuremath{\goliath, \david})\xspace}
\newcommand{\Mdeliver}{(\textsf{deliver}, \textsf{sid}, \ensuremath{\goliath, \david})\xspace}
\newcommand{\Mready}{(\Sready, \textsf{sid}, \ensuremath{\goliath, \david})\xspace}
\newcommand{\Mexecute}{(\textsf{execute}, \textsf{sid}, \ensuremath{\goliath, \david, \inp})\xspace}
\newcommand{\Mout}[1]{(\textsf{output}, \textsf{sid}, \ensuremath{\goliath, \david, #1})\xspace}
\newcommand{\Mincom}{(\textsf{communication}, \textsf{sid}, \ensuremath{\goliath, \david, m})\xspace}
\newcommand{\Mchoice}[1]{(\textsf{choice}, \textsf{sid}, \ensuremath{\goliath, \david, #1})\xspace}
\newcommand{\Mswitchoff}{(\textsf{switch off}, \textsf{sid}, \ensuremath{\goliath, \david})\xspace}
\newcommand{\Mswitchon}{(\textsf{switch on}, \textsf{sid}, \ensuremath{\goliath, \david})\xspace}
\newcommand{\Ftwo}{\ensuremath{\mathbb{F}_2}\xspace}
\newcommand{\Tout}[1]{(\textsf{output}, \ensuremath{#1})\xspace}
\newcommand{\Tchoice}[1]{(\textsf{choice}, \ensuremath{#1})\xspace}
\newcommand{\Tcommit}[1]{(\textsf{commit}, \ensuremath{#1})\xspace}
\newcommand{\Tchallenge}[1]{(\textsf{challenge}, \ensuremath{#1})\xspace}
\newcommand{\Tmatrix}[1]{(\textsf{matrix choice}, \ensuremath{#1})\xspace}
\newcommand{\Tvector}[1]{(\textsf{vector choice}, \ensuremath{#1})\xspace}
\newcommand{\Tcom}[1]{(\textsf{commitment}, \ensuremath{#1})\xspace}
\newcommand{\TMcom}[1]{(\textsf{message commitment}, \ensuremath{#1})\xspace}
\newcommand{\TMrev}[1]{(\textsf{reveal message}, \ensuremath{#1})\xspace}
\newcommand{\TMopen}[1]{(\textsf{opening}, \ensuremath{#1})\xspace}

\newcommand{\TMshares}[1]{(\textsf{shares opening}, \ensuremath{#1})\xspace}

\begin{document}

% Title, authors, abstract and keywords are in the abstract file
\def\titletext{Weakening the Isolation Assumption of Tamper-proof Hardware Tokens}
\def\shorttitletext{Weakening the Isolation Assumption of Tamper-proof Hardware Tokens}

\def\abstracttext{Recent results have shown the usefulness of tamper-proof hardware tokens as a setup assumption 
for building UC-secure two-party computation protocols, thus providing broad security guarantees and allowing 
the use of such protocols as buildings blocks in the modular design of complex cryptography protocols. All these 
works have in common that they assume the tokens to be completely isolated from their creator, but this is a 
strong assumption. In this work we investigate the feasibility of cryptographic protocols in the setting where the isolation of the hardware token is weakened. 

We consider two cases: (1) the token can relay messages to its creator, or (2) the creator can send messages to the token after it is sent to the receiver. We provide a detailed characterization for both settings, presenting both impossibilities and information-theoretically secure solutions.
}

\def\keywordstext{\textbf{Keywords\ifnum\llncs=1{.}\else{:}\fi}
Hardware Tokens, Isolation Assumption, UC security, One-Time Memory, Oblivious Transfer.}

\def\dateinfo{November 2014}

\ifnum\llncs=1
  \def\pubinfo{}
\else
  \def\pubinfo{A conference version appears at ICITS 2015. This is the full version.}
\fi

\def\Dowsley{Rafael Dowsley}
\def\DowsleyEmail{{\tt rafael.dowsley@kit.edu}}
\def\DowsleyWebpage{{\tt }}
\def\DowsleyDept{Institute of Theoretical Informatics}
\def\DowsleyCompany{Karlsruhe Institute of Technology}
\def\DowsleyAddress{Am Fasanengarten 5, Geb. 50.34, 76131 Karlsruhe, Germany}

\def\MuellerQuade{J\"{o}rn M\"{u}ller-Quade}
\def\MuellerQuadeEmail{{\tt mueller-quade@kit.edu}}
\def\MuellerQuadeWebpage{{\tt }}
\def\MuellerQuadeDept{Institute of Theoretical Informatics}
\def\MuellerQuadeCompany{Karlsruhe Institute of Technology}
\def\MuellerQuadeAddress{Am Fasanengarten 5, Geb. 50.34, 76131 Karlsruhe, Germany}

\def\Nilges{Tobias Nilges}
\def\NilgesEmail{{\tt tobias.nilges@kit.edu}}
\def\NilgesWebpage{{\tt }}
\def\NilgesDept{Institute of Theoretical Informatics}
\def\NilgesCompany{Karlsruhe Institute of Technology}
\def\NilgesAddress{Am Fasanengarten 5, Geb. 50.34, 76131 Karlsruhe, Germany}

\def\DowsleyThanks{
\DowsleyDept,
\DowsleyCompany,
\DowsleyAddress.
Email: {\tt \DowsleyEmail}.
%URL: {\tt \DowsleyWebpage}.
}

\def\MuellerQuadeThanks{
\MuellerQuadeDept,
\MuellerQuadeCompany,
\MuellerQuadeAddress.
Email: {\tt \MuellerQuadeEmail}.
%URL: {\tt \MuellerQuadeWebpage}.
}

\def\NilgesThanks{
\NilgesDept,
\NilgesCompany,
\NilgesAddress.
Email: {\tt \NilgesEmail}.
%URL: {\tt \NilgesWebpage}.
}

\def\DowsleyInstitute{
\DowsleyDept \\
\DowsleyCompany \\
\DowsleyAddress \\
\email{\DowsleyEmail}%\\
\ifnum\instcompact=1
; \MuellerQuadeEmail ; \NilgesEmail \fi
%\email{\DowsleyWebpage}
}

\def\MuellerQuadeInstitute{
\MuellerQuadeDept \\
\MuellerQuadeCompany \\
\MuellerQuadeAddress \\
\email{\MuellerQuadeEmail}%\\
%\email{\MuellerQuadeWebpage}
}

\def\NilgesInstitute{
\NilgesDept \\
\NilgesCompany \\
\NilgesAddress \\
\email{\NilgesEmail} %\\
%\email{\NilgesWebpage}
}

\def\DowsleyThanksInline{
\DowsleyDept,
\DowsleyCompany \\
\DowsleyAddress \\
{\tt \DowsleyEmail}%\\
\ifnum\instcompact=1
; \MuellerQuadeEmail ; \NilgesEmail \fi
%{\tt \DowsleyWebpage}
}

\def\MuellerQuadeThanksInline{
\MuellerQuadeDept,
\MuellerQuadeCompany \\
\MuellerQuadeAddress \\
{\tt \MuellerQuadeEmail}%\\
%{\tt \MuellerQuadeWebpage}
}

\def\NilgesThanksInline{
\NilgesDept,
\NilgesCompany \\
\NilgesAddress \\
{\tt \NilgesEmail} %\\
%{\tt \NilgesWebpage}
}

\ifnum\llncs=1
  \title{\titletext}
  \titlerunning{\shorttitletext}
  \ifnum\anonymous=1
    \author{}
    \authorrunning{}
    \institute{}
  \else
    \author{
      \ifnum\instcompact=0 
        \Dowsley\inst{1} \and
        \MuellerQuade\inst{2} \and
        \Nilges\inst{3}
      \else
        \Dowsley \and
        \MuellerQuade \and
        \Nilges
      \fi        
    }
    \institute{    
      \ifnum\instcompact=0 
        \DowsleyInstitute \and
        \MuellerQuadeInstitute \and
        \NilgesInstitute
      \else
        \DowsleyInstitute 
      \fi
    }
  \fi
  \maketitle

  \ifnum\savespace=1{\vspace{-0.5in}}\fi
  \begin{abstract}
  \abstracttext
  
  \vspace{0.05in}
  \noindent \keywordstext
  \end{abstract}
\else
  \ifnum\anonymous=1
    \date{}
    \appearsin{}
    \title{\textbf{\titletext}\vspace{-0.3in}}
    \author{}
    \maketitle
    \vspace{-0.4in}
    \begin{abstract}
    \abstracttext
    
    \ifnum\showkeywords=1 
      \vspace{0.15in}
      \noindent \keywordstext 
    \fi
    \end{abstract}\vspace{0.2in}
    \vspace{-0.2in}
  \else
    \appearsin{\pubinfo}
    \title{{\bf \titletext}\vspace{0.2in}}
    \author{
      \ifnum\instplace=0   
         {\sc \Dowsley}\thanks{~\DowsleyThanks} \and
         {\sc \MuellerQuade}\thanks{~\MuellerQuadeThanks} \and
         {\sc \Nilges}\thanks{~\NilgesThanks}
      \else
        \ifnum\instcompact=0 
          {\Dowsley$^{1}$} \and
          {\MuellerQuade$^{2}$} \and
          {\Nilges$^{3}$} \and
        \else
          {\Dowsley} \and
          {\MuellerQuade} \and
          {\Nilges}
        \fi
      \fi
    }
    \ifnum\showdate=1
      \date{\protect\vspace{0.2in} \dateinfo}
    \else
      \date{}
    \fi
    \maketitle

    \ifnum\instplace=1
      \vspace{-0.2in}
      \begin{center}
      \begin{small}
       \ifnum\instcompact=0
         $^{1}$~\DowsleyThanksInline \\
         $^{2}$~\MuellerQuadeThanksInline \\
         $^{3}$~\NilgesThanksInline \\
       \else
         \DowsleyThanksInline \\
       \fi
      \end{small}
      \end{center}
      \vspace{0.2in}
    \fi
    \begin{abstract}
    \abstracttext
    
    \ifnum\showkeywords=1 
      \vspace{0.15in}
      \noindent \keywordstext 
    \fi
    \end{abstract}
    \thispagestyle{empty}\newpage\setcounter{page}{1}
    \ifnum\addtoc=1
      \pagenumbering{roman}
      \newpage
      \def\baselinestretch{0.9}
      \tableofcontents
      \def\baselinestretch{1.0}
      \newpage
      \pagenumbering{arabic}
      \setcounter{page}{1}
    \fi
  \fi
\fi

\section{Introduction}

Tamper-proof hardware tokens are a valuable resource for designing cryptographic protocols. It was shown in a series of recent papers that tamper-proof hardware tokens can be used as a cryptographic setup assumption to obtain Universally Composable (UC)~\cite{FOCS:Canetti01} secure two-party computation protocols~\cite{EC:Katz07,EC:MovSeg08,TCC:GISVW10,C:GIMS10,TCC:DotKraMul11}, thus achieving solutions that are secure according to one of the most stringent cryptographic models and can be used as buildings blocks in the modular design of complex cryptography protocols. D\"{o}ttling et al.~\cite{TCC:DotKraMul11} showed that even a single tamper-proof hardware token generated by one of the mutually distrusting parties 
%(and thus not trusted by the other party) 
is enough to obtain information-theoretical security in the UC framework.

All these works have in common that the tokens are assumed to be completely isolated from their creator. In light of recent events this assumption becomes questionable at the least, apart from the fact that the tokens could contain internal clocks, which can be exploited in conjunction with the activation time to send information into the device (or to make the abort behavior dependent on the activation time, which is not modeled in the UC framework). We highlight that this problem lies skew to leakage and side-channel attacks, e.g.\ \cite{AC:BCGHKR11,ICALP:PraSahWad14}, where a malicious token receiver tries to extract some of the contents of the token, i.e.\ the \emph{tamper-resilience assumption} is weakened. In contrast, we consider a weakened \emph{isolation assumption}. A similar scenario was studied by Damg\aa rd et al.\ \cite{TCC:DamNieWic09}, but only for a bandwidth-restricted channel and computational security. They showed that a partial physical separation of parties, e.g.\ in a token with a low-bandwidth covert channel, allows to perform UC-secure multiparty computation under standard cryptographic assumptions. 

We consider an \emph{unrestricted} channel and information-theoretical security. In this scenario, communication in both directions between the token and its creator without any restriction obviously renders the token useless as a setup assumption. Thus, there remain two different kinds of communication that can be considered to weaken the isolation assumption: either the tokens' creator can send messages to the tokens, or the tokens can send messages to their creator. While we deem the first case to be more realistic, we consider both cases. We emphasize that these one-way channels are available only for malicious parties and thus are not used by the honest parties during the protocol execution. This scenario is not directly comparable with the one by Damg\aa rd et al.\ \cite{TCC:DamNieWic09}, since here a broadband communication channel is available, but it is only one-way.
This leads to the following question:

\begin{center}
	\emph{Is it possible to obtain UC-secure protocols even if there exists a broadband one-way communication channel between the tokens and their creator?}
\end{center}

In this work, we provide a broad characterization from a feasibility standpoint for both malicious incoming and outgoing communication between the tokens and their creator. For our solutions, we only require that one party can create hardware tokens. We thus call this party Goliath, while the receiver of the token is called David and cannot create tokens of its own. 

In more detail, we show that with one-way channels into the tokens, it is possible to basically use the One-Time Memory (OTM) protocol using two tokens of D\"{o}ttling et al. \cite{TCC:DotKraMul11} to obtain an information-theoretically UC-secure OTM with aborts (i.e., a malicious token creator can change the abort behavior of the token at runtime, which is unavoidable if one-way channels into the tokens are available) and we also provide a computationally UC-secure OTM protocol from a single token. Additionally, it is possible to obtain information-theoretically UC-secure Oblivious Transfer (OT) from a single hardware token. We prove an impossibility result for unconditionally secure OTM with a single token.

Concerning one-way  channels from the tokens to their creator, we show that it is impossible to obtain even information-theoretically secure OT. We provide an information-theoretically UC-secure commitment scheme, which can then be used to obtain a computationally UC-secure OTM protocol with known techniques \cite{C:PeiVaiWat08}.  

\textbf{Further related work.} 
Apart from the model of tamper-proof hardware as formalized by Katz~\cite{EC:Katz07}, also weaker models such as resettable hardware tokens were proposed, e.g.\ \cite{TCC:GISVW10}. With resettable hardware, it is not possible to obtain information-theoretically secure oblivious transfer~\cite{C:GIMS10}, while commitments are still possible~\cite{C:GIMS10,AC:DamSca13}. Thus, the main focus of this research direction are efficient protocols based on computational assumptions while minimizing the amount of communication and tokens~\cite{EC:ChaGoySah08,TCC:GISVW10,TCC:Kolesnikov10,AC:BCGHKR11,TCC:DMMN13,TCC:CKSYZ14}. 
Further results about hardware tokens can be founded in \cite{C:ChaPed92,C:Brands93,EC:CraPed93,C:IshSahWag03,TCC:GLMMR04,TCC:AAGPR14}.

Another UC hardware setup assumption are physically uncloneable functions (PUFs)~\cite{pap01,C:BFSK11,EC:OSVW13}, which have recently gained increasing interest. It was shown that PUFs can be used to achieve oblivious transfer~\cite{rue10} and UC-secure commitments~\cite{AC:DamSca13}. However, if the PUFs can be created maliciously, oblivious transfer is impossible~\cite{C:DFKLS14}.

\section{Preliminaries}

\subsection{Notation}
We use standard information-theoretic measures: by $H(\cdot)$ we denote Shannon entropy, $H(\cdot | \cdot)$ denotes conditional entropy and $I(\cdot;\cdot)$ denotes the mutual information. Let in the following $\secpar$ denote a security parameter. We use the cryptographic standard notions of negligible functions, as well as computational\slash statistical\slash perfect indistinguishability.

\subsection{Model}

We state and prove our results in the Universal Composability (UC) framework of Canetti~\cite{FOCS:Canetti01} that allows for arbitrary composition of protocols. In this framework an ideal functionality $\mathcal{F}$ that captures the desired security requirements has to be modeled. A protocol $\Pi$ that is supposed to instantiate $\mathcal{F}$ runs in the real world, where an adversary $\mathcal{A}$ can corrupt protocol parties. To prove the UC-security of $\Pi$, it has to be shown that there exists a simulator $\mathcal{S}$ that only interacts with the ideal functionality and simulates the behavior of any $\mathcal{A}$ in such a way that any environment $\mathcal{Z}$ that is plugged either into the real protocol or the simulated protocol cannot distinguish the real protocol run of $\Pi$ from a simulated one.\footnote{In the case of computational security we allow the simulator to be expected polynomial time.}
For our results we assume static corruption, i.e.\ the adversary cannot adaptively corrupt protocol parties.

\paragraph{Target Functionalities.} Ideally one would like to use tamper-proof hardware tokens to realize 
One-Time Memory (OTM)~\cite{C:GolKalRot08}, as in the case where
the token is modeled as being completely isolated from its creator~\cite{TCC:DotKraMul11}. See \figref{fig-otm} for the OTM
functionality definition. This primitive resembles oblivious transfer, but
the receiver can make his choice at any point in time and the sender
is not notified about this event. OTM allows to build One-Time 
Programs~\cite{C:GolKalRot08,TCC:GISVW10}.

\paragraph{Impossibility of Realizing OTMs.} Note that in the hybrid
execution with a token and a channel into the token, a dishonest sender \goliath has the
ability to send an abortion message to \token at any time,
thus changing its abort behavior. In
the ideal execution on the other hand, once the OTM functionality goes
to the $\Sready$ state, it is not possible to change its output/abort behavior
anymore. Therefore it is not possible to realize the OTM functionality 
based on tokens that can receive communication from a malicious \goliath.

\begin{figure}[tb]
\begin{small}
\begin{tabular}{|p{\textwidth}|}
\hline
\begin{center}
\textbf{Functionality} $\calF^{\text{OTM}}$
\smallskip
\end{center}\\
Parametrized by a security parameter $\secpar$. The variable \Fstate is initialized with $\Swait$.\\
\\
\textbf{Creation.} Upon receiving a message \Mcreate{s_0,s_1} from 
\goliath verify if $\Fstate=\Swait$ and $s_0, s_1 \in \bits^\secpar$; else abort. 
Next, set $\Fstate \gets \Ssent$, store $(\sid, \goliath, \david, 
s_0,s_1)$ and send the message \Mcreated to the adversary.\\
\\
\textbf{Deliver.} Upon receiving a message \Mdeliver from the adversary, verify that 
$\Fstate=\Ssent$; else abort. Next, set $\Fstate \gets \Sready$, 
and send \Mready to \david.\\
\\
\textbf{Choice.} Upon receiving a message \Mchoice{c} from \david 
check if $\Fstate=\Sready$; else abort. Next, set $\Fstate \gets \Sdead$ 
and send \Mout{s_c} to \david.\\
\\
\hline
\end{tabular}
\smallskip
\caption{The One-Time Memory functionality.}\label{fig-otm}
\end{small}
\end{figure}

\begin{figure}[tb]
\begin{small}
\begin{tabular}{|p{\textwidth}|}
\hline
\begin{center}
\textbf{Functionality} $\calF^{\text{OTM-with-Abort}}$
\smallskip
\end{center}\\
Parametrized by a security parameter $\secpar$.
The variable \Fstate is initialized with $\Swait$ and $\Fabort$ with $\top$. If any message other than
\Mswitchon is received while $\Fabort=\bot$, the functionality aborts.\\
\\
\textbf{Creation.} Upon receiving a message \Mcreate{s_0,s_1} from 
\goliath verify if $\Fstate=\Swait$ and $s_0, s_1 \in \bits^\secpar$; else abort. 
Next, set $\Fstate \gets \Ssent$, store $(\sid, \goliath, \david, 
s_0,s_1)$ and send the message \Mcreated to the adversary.\\
\\
\textbf{Overwrite.} Upon receiving a message \Moverwrite{s_0',s_1'} from 
$\adv$ verify if $\Fstate=\Ssent$ and $s_0', s_1' \in \bits^\secpar$; else abort. 
Set $s_0 \gets s_0' \next s_1 \gets s_1'$.\\
\\
\textbf{Deliver.} Upon receiving a message \Mdeliver from the adversary, verify that 
$\Fstate=\Ssent$; else abort. Next, set $\Fstate \gets \Sready$, 
and send \Mready to \david.\\
\\
\textbf{Choice.} Upon receiving a message \Mchoice{c} from \david 
check if $\Fstate=\Sready$; else abort. Next, set $\Fstate \gets \Sdead$ 
and send \Mout{s_c} to \david.\\
\\
\textbf{Switch Off.} Upon receiving a message \Mswitchoff from $\adv$ set $\Fabort \gets \bot$.\\ 
\\
\textbf{Switch On.} Upon receiving a message \Mswitchon from $\adv$, set $\Fabort \gets \top$.\\ 
\\
\hline
\end{tabular}
\smallskip
\caption{The One-Time Memory with Abort functionality.}\label{fig-otmwa}
\end{small}
\end{figure}

\paragraph{OTM with Abort.} Given the above fact that online changes
in the abort behavior are inherent in the setting with one-way communication
into the token, we introduce an OTM functionality 
with abort, see \figref{fig-otmwa}. For such a functionality, there is an initial 
delivering phase after which the adversary can only let the execution proceed correctly
or switch off the functionality whenever he wants (independent of David inputs); but he cannot change the values stored in the functionality.

\section{The Case of Incoming Communication}

We first show that the existing solution of D\"{o}ttling, Kraschewski and M\"{u}ller-Quade~\cite{TCC:DotKraMul11} for OTM with 2 tokens can be modified to UC-realize OTM with abort. Then we show that using a single token, it is impossible to obtain an information-theoretically secure OTM protocol, if Goliath can send messages to the token. We sketch how a information-theoretically UC-secure OT protocol from a single token can be obtained and give a construction of a compuationally UC-secure OTM protocol from a single hardware token.

The formalization of the ideal functionality for stateful tamper-proof hardware tokens in this section uses a wrapper functionality as in the previous works~\cite{EC:Katz07,EC:MovSeg08,TCC:DotKraMul11}, but as one-way communication from the token issuer to the token is now allowed, the wrapper functionality needs to be modified to capture this fact. A sender \goliath (Goliath) provides as input to $\calF_{\text{wrap-owc}}^{\text{stateful}}$ a deterministic Turing machine \token (the token). Note that stateful tokens can be hard-coded with sufficiently long randomness tapes. The receiver \david (David) can query $\calF_{\text{wrap-owc}}^{\text{stateful}}$ to run \token with inputs of his choice and receives the output produced by the token. The current state of \token is stored between consecutive queries. 
In addition, and in order to capture the one-way communication property, we add the possibility of Goliath sending messages to the token, in which case \token is run on the received string and changes to a new state. The complete description of the functionality is shown in \figref{fig-wrap}. This model captures the fact that on the one hand the token cannot send messages to its creator, and on the other hand David cannot access the code or the internal state of \token.

\begin{figure}[tb]
\begin{small}
\begin{tabular}{|p{\textwidth}|}
\hline
\begin{center}
\textbf{Functionality} $\calF_{\text{wrap-owc}}^{\text{stateful}}$
\smallskip
\end{center}\\
Parametrized by a security parameter $\secpar$ and a polynomial 
upper bound on the runtime $t(\cdot)$. The variable \Fstate is initialized with $\Swait$.\\
\\
\textbf{Creation.} Upon receiving a message \Mcreate{\token} from \goliath
where \token is a deterministic Turing machine, verify if $\Fstate=\Swait$; else 
ignore the input. Next, store (\sid, \goliath, \david, \token,$\Tstate$) where 
$\Tstate$ is the initial state of \token, set $\Fstate \gets \Ssent$ and 
send the message \Mcreated to the adversary.\\
\\
\textbf{Deliver.} Upon receiving a message \Mdeliver from the adversary, verify that $\Fstate=\Ssent$; else 
ignore the input. Next, set $\Fstate \gets \Sready$,  and send \Mready to \david.\\
\\
\textbf{Execution.} Upon receiving a message \Mexecute from \david 
where \inp is an input, check if $\Fstate=\Sready$ and if it is, 
then run $\token(\Tstate,\inp)$ for at most $t(\secpar)$ steps. Save 
the new state of \token in \Tstate, read the output \out from its output tape and 
send \Mout{\out} to \david.\\
\\
\textbf{Incoming Communication.} Upon receiving a message \Mincom
from \adv, run $\token(\Tstate,m)$ for at most 
$t(\secpar)$ steps. Save the new state of \token.\\
\\
\hline
\end{tabular}
\smallskip
\caption{The wrapper functionality allowing one-way communication.}\label{fig-wrap}
\end{small}
\end{figure}

\subsection{Unconditionally Secure OTM with Two Tokens}

\begin{figure}[tb]
\begin{small}
\begin{tabular}{|p{\textwidth}|}
\hline
\begin{center}
\textbf{Token - Random Values} $\token_{\text{Random}}$
\smallskip
\end{center}\\
Parametrized by a security parameter $\secpar$. The token is 
hardwired with a random vector $a \getsr \Ftwo^{2\secpar}$ and 
a random matrix $B \getsr \Ftwo^{2\secpar \times 2\secpar}$.
It is initialized with state $\Tstate =\Sready$.\\
\\
\textbf{Output.} Upon receiving a message \Tchoice{z} from \david 
check if $\Tstate=\Sready$; else abort. Next, set $\Tstate \gets \Sdead$,
compute $V \gets a \otimes z + B$ and send the message \Tout{V} to \david.\\
\\
\hline
\end{tabular}
\smallskip
\caption{The first token, which only contains random values.}\label{fig-tok1}
\end{small}
\end{figure}

\begin{figure}[tb]
\begin{small}
\begin{tabular}{|p{\textwidth}|}
\hline
\begin{center}
\textbf{Token - Inputs} $\token_{\text{Inputs}}$
\smallskip
\end{center}\\
Parametrized by a security parameter $\secpar$. The token is 
initialized with Goliath's inputs $s_0, s_1$, and the vector $a$ and
matrix $B$ that are used by $\token_{\text{Random}}$.
It is initialized in state $\Tstate =\Sready$.\\
\\

\textbf{Matrix Choice.} Upon receiving a message \Tmatrix{C} from \david 
check if $\Tstate=\Sready$ and $C \in \Ftwo^{\secpar \times 2 \secpar}$; 
else abort. Next, compute a matrix $G \in \Ftwo^{\secpar \times 2 \secpar}$ that is complementary 
to $C$ (i.e., $G$ is determined by $\secpar$ vectors of length $2 \secpar$ which are linearly independent
and $G$ spans a subspace of the kernel of $C$), and also compute $\tilde{a} \gets Ca$, $\tilde{B}\gets CB$. Set $\Tstate \gets \Scom$
and send the message \Tcom{G,\tilde{a},\tilde{B}} to \david.\\
\\

\textbf{Ciphertexts.} Upon receiving a message \Tvector{h} from \david 
check if $\Tstate=\Scom$ and $h \in \Ftwo^{2 \secpar} \setminus \{0\}$; 
else abort. Next, compute $\tilde{s_0} \gets s_0 + GBh$ and 
$\tilde{s_1} \gets s_1 + GBh+Ga$, set $\Tstate \gets \Sdead$ and send 
the message \Tout{\tilde{s_0},\tilde{s_1}} to \david.\\
\\
\hline
\end{tabular}
\smallskip
\caption{The second token, which stores Goliath's inputs.}\label{fig-tok2}
\end{small}
\end{figure}

\begin{figure}[tb]
\begin{small}
\begin{tabular}{|p{\textwidth}|}
\hline
\begin{center}
\textbf{Protocol}
\smallskip
\end{center}\\
Parametrized by a security parameter $\secpar$.\\
\\
\textbf{Deliver.} $\david$ waits until $\goliath$ send the tokens $\token_{\text{Random}}$ 
and $\token_{\text{Inputs}}$. Then he chooses a random matrix 
$C \in \Ftwo^{\secpar \times 2 \secpar}$ and sends the message 
\Tmatrix{C} to $\token_{\text{Inputs}}$ in order to get the answer \Tcom{G,\tilde{a},\tilde{B}}. 
After that, $\david$ picks a random vector $h \in \Ftwo^{2 \secpar} \setminus \{0\}$
and sends the message \Tvector{h} to $\token_{\text{Inputs}}$ in order 
to get the output \Tout{\tilde{s_0},\tilde{s_1}}.\\
\\
\textbf{Choice Phase.} When $\david$ gets his input $c \in \Ftwo$, he chooses $z \getsr \Ftwo^{2 \secpar}$
such that $z^Th=c$ and sends the message \Tchoice{z} to $\token_{\text{Random}}$
to get the output \Tout{V}. Then $\david$ checks if $CV=\tilde{a}z^T+\tilde{B}$. If it is not,
$\david$ aborts; otherwise, he outputs $s_c=\tilde{s_c}+GVh$.\\
\\
\hline
\end{tabular}
\smallskip
\caption{The unconditionally secure protocol that realizes $\calF^{\text{OTM-with-Abort}}$.}\label{fig-prot}
\end{small}
\end{figure}

Our solution is to use the non-interactive 
version of the protocol due to D\"{o}ttling, Kraschewski and 
M\"{u}ller-Quade~\cite{TCC:DotKraMul11}. The only function of Goliath 
in this protocol is creating the two tokens and sending them to 
David. David, on the other hand, interacts with both 
tokens in order to obtain his output and to check the correctness
of the protocol execution. Intuitively, one of the tokens is used to generate 
a commitment to the input values and to send the input values encrypted using 
one-time pads. The second token only contains a random affine function which
can be evaluated only a single time and allows David to recover the 
one-time pad key corresponding to one of the inputs.
The specifications of the tokens can be found in \figref{fig-tok1} 
and \figref{fig-tok2}. In the protocol David initially interacts with 
the token which has the inputs in order to obtain the commitments 
and the ciphertexts. After this point David considers the OTM 
as delivered. Then, whenever he wants to choose the input to be received, 
he simply queries the token that has the affine function 
on the appropriate input and obtains the one-time pad that he needs
in order to recover his desired value. The description of the protocol is presented in
\figref{fig-prot}.

The fact that the protocol securely realizes $\calF^{\text{OTM-with-Abort}}$
follows from a straightforward modification of the original security proof by D\"{o}ttling et al.~\cite{TCC:DotKraMul11},
which considered the same protocol but with isolated tokens and proved that it realizes 
$\calF^{\text{OTM}}$ (i.e., without aborts) in such scenario.

\begin{theorem}
In the model where a malicious Goliath is allowed to send messages to the token, the protocol presented in \figref{fig-prot} UC-realizes
the functionality $\calF^{\text{OTM-with-Abort}}$ with statistical security against a corrupted Goliath and perfect security against a corrupted David.
\end{theorem}

\begin{proof} (Sketch) 
The correctness as well as the security against a corrupted David follow directly from D\"{o}ttling's et al. 
proof of security. In the case of the security against a corrupted Goliath, note that the OTM is considered 
delivered at the point in which David has received $(G,\tilde{a},\tilde{B},\tilde{s_0},\tilde{s_1})$ 
from $\token_{\text{Inputs}}$. From that point on, $\token_{\text{Inputs}}$ does not participate 
in the protocol anymore and it cannot send messages to the outside world. Hence 
neither Goliath nor $\token_{\text{Random}}$ know the matrix $C$ which is used for the 
commitments, so they can cheat in the commitment's opening phase only with negligible probability. Both of them
also do not know the value $h$, which is necessary together with $z$ in order to determine David's
input $x$. So the proof proceeds as in~\cite{TCC:DotKraMul11},
the only difference here is that Goliath can still send messages to $\token_{\text{Random}}$ at any point, and thus
he can modify the abort behavior. This can be dealt with by running D\"{o}ttling's et al. procedure to verify whether the token
is going to abort or not (i.e., running a copy of the token in its current state with random inputs)
after each incoming message from Goliath to the token. If the simulator notices that the abort behavior 
changed, he can make the appropriate change in $\calF^{\text{OTM-with-Abort}}$ by using the Switch Off/Switch On
commands.
\end{proof}

\paragraph{Sequential OTM with Abort.} As done by D\"{o}ttling et al.~\cite{TCC:DotKraMul11} for the OTM functionality, 
it is also possible to define a sequential version of the OTM-with-Abort functionality where there are many pairs of Goliath's 
inputs (i.e., there are multiple stages) which can only be queried sequentially by David. The functionality only needs to be modified to take
pairs of inputs which can be queried sequentially by David and to allow an adversary to specify which stages are active/inactive at 
any time (if an inactive stage is queried by David, then the functionality aborts). In this case the two token solution of 
D\"{o}ttling et al.~\cite{TCC:DotKraMul11} for sequential OTMs can be used. The security proof would be a straightforward 
modification of D\"{o}ttling et al.'s proof in the same line as done above.

\subsection{Impossibility of Unconditionally Secure OTM from a Single Token}

\begin{lemma} 
Assume that there is only one token and that a malicious token is not computationally bounded.
If a malicious Goliath is allowed to send messages to the token, then there is no protocol $\Pi$ that realizes 
OTM with information-theoretic security from this single token.
\end{lemma}

\begin{proof} 
For the sake of contradiction assume that a correct and information-theoretically secure OTM protocol $\Pi$ from a single 
stateful token exists. Assume that the parties' inputs are chosen as $s_0,s_1 \getsr \bits^\lambda$ and $c \getsr \bits$. 
The sender's privacy of the OTM protocol should hold, i.e.\
\begin{eqnarray*} 
I(\mathsf{view}_\david; s_{1-c}) \leq\varepsilon	& \Leftrightarrow & H(s_{1-c})-H(s_{1-c}\cmid \mathsf{view}_\david) \leq\varepsilon \\
	& \Leftrightarrow &H(s_{1-c}\cmid \mathsf{view}_\david) \geq \lambda- \varepsilon,
\end{eqnarray*}
where $\mathsf{view}_\david$ is David's view of the protocol execution and $\varepsilon$ is a function that is negligible in the security parameter.

By definition of the OTM functionality David can choose his input $c$ at any time after he receives the token and Goliath should not 
learn when David queried the OTM functionality. So David can choose his input $c$ at a point in the future far after receiving the token,
when all initial communication between the parties is already finished, and then he interacts with the token to receive $s_c$. 
But then, at the moment right before David's choice $c$ is made, its entropy is still 1 from the point of view of all parties. Therefore, due to the sender's privacy, 
at this point it should hold that 
$$H(s_{0}\cmid \mathsf{view}'_\david) \geq \lambda - \varepsilon$$
and
$$H(s_{1}\cmid \mathsf{view}'_\david) \geq \lambda - \varepsilon,$$
where $\mathsf{view}'_\david$ is David's view of the protocol execution until this point. But if a malicious Goliath is allowed to 
send messages to the token, he can forward his complete view to the token. The token then gets to know all protocol interactions 
so far and due to the correctness of the OTM protocol (i.e., it should work for any pair of inputs in $\bits^\lambda$) he is able, for almost any $s'_c \in \bits^\lambda$, to find a strategy to follow
for the rest of the protocol that makes David accept $s'_c$. Hence the values $s_0$ and $s_1$ are not fixed up to the point 
when David inputs $c$. But in the OTM functionality the values $s_0$ and $s_1$ are fixed once it is sent, and thus we get a contradiction.
\end{proof}

%Assume there exists a protocol $\Pi$ with 2 or more stages from a single stateful hardware token that UC-realizes $\mathcal{F}_{seq-OTM}$ even if Goliath is allowed to send messages to the token. Due to the definition of $\mathcal{F}_{seq-OTM}$, we can further assume that the last message is interchanged between David and the token, i.e.\ Goliath does not know when the token is queried. Using the channel into the token, Goliath can send its complete view after each interaction with David to the token. This means that at the point when David queries the token for the last message, the token has a complete view of all the interaction between David and Goliath and its own interaction with David. This scenario corresponds to a token view of a non-interactive OTM protocol.
%
%Now we can use the impossibility result of Döttling et al.\cite{dkmq11}, which states that an UC-secure OTM protocol from a single stateful hardware token with 2 or more stages needs interaction.

\subsection{Unconditionally Secure OT with a Single Token} 
D\"{o}ttling et al.~\cite{TCC:DotKraMul11} also presented an unconditionally secure solution with one
token only, in which the interactions which are performed between David and $\token_{\text{Inputs}}$ in the previously described protocol are
instead performed between David and Goliath in an initial interactive phase that is used to send the commitments and the ciphertexts. 
Note that such a version of the protocol would not be secure in the setting where one-way communication is allowed
into the token since Goliath could simply forward the matrix $C$ to $\token_{\text{Random}}$, which would then
be able to open the commitments to any value and thus be able to change the outputs at any time. 
But we should mention that it is possible to obtain an oblivious transfer protocol with only one token by letting the single token 
act like $\token_{\text{Inputs}}$ in the above protocol and letting the interactions between David and $\token_{\text{Random}}$
be replaced by identical interactions between David and Goliath. The proof of security would follow in the same line as before since 
Goliath would never get to know $C$ and $h$. Note that the drawback of having to know the OT inputs before
sending the token can be easily overcome by performing the OTs with random inputs and derandomizing them afterwards.

\subsection{Computationally Secure OTM from a Single Token}\label{sec:comp}

\begin{figure}[tb]
\begin{small}
\begin{tabular}{|p{\textwidth}|}
\hline
\begin{center}
\textbf{Token} $\token$
\smallskip
\end{center}\\
Parametrized by a security parameter $\secpar$. The token is hardwired with the shares 
$(v_{i,0},v_{i,1})$ for $i=1, \ldots, \secpar$, the inputs $s_0,s_1$ and Goliath's public key
$pk$. It is initialized with state $\Tstate =\Sready$ and $j=0$.\\
\\
\textbf{Message Commitment.} Upon receiving a message \Tchallenge{k_j} from \david 
check if $\Tstate=\Sready$ and $k_j$ is a bit; else abort Set $j \gets j + 1$. If $j=\secpar$, then 
set $\Tstate \gets \SMcom$. Send the message \TMcom{v_{j,k_j}} to \david.\\
\\

\textbf{Inputs Commitment.} Upon receiving a message \Tcommit{\mathsf{crs},\sigma} from \david 
check if $\Tstate=\SMcom$ and if $\sigma$ is a valid signature of $\goliath$ on $\mathsf{crs}$; else abort. 
Set $\Tstate \gets \SIcom$. Commit to the values $s_0,s_1$ using a computationally 
UC-secure commitment protocol that uses the common reference string $\mathsf{crs}$ and send the 
commitments to \david. Let $d_0,d_1$ denote the information to open the commitments.\\
\\

\textbf{Output.} Upon receiving the message \textsf{output} from \david 
check if $\Tstate=\SIcom$; else abort. Next, set $\Tstate \gets \Sdead$ and 
execute with \david a computationally UC-secure oblivious transfer protocol using 
the common reference string $\mathsf{crs}$ and with inputs $(s_0 \| d_0,s_1\| d_1)$.\\ 
\\
\hline
\end{tabular}
\smallskip
\caption{The token for a computationally secure OTM protocol with a single token.}\label{fig-tok3}
\end{small}
\end{figure}

If one considers the scenario where only one token is available, it is possible to obtain a protocol that realizes 
$\calF^{\text{OTM-with-Abort}}$ with computational security. The idea is to compute as an initial step (i.e.\,during the delivery phase) 
the commitment functionality by using the token and interactions
between Goliath and David. With access to this commitment functionality it is possible to 
obtain a common reference string between David and the token\footnote{The common reference string 
is actually obtained by Goliath and David, but can be forwarded from Goliath to the token via David 
by using a digital signature to ensure that the value that the token obtains is exactly the same 
one that Goliath sent.}, which in turn allows to run a computationally 
secure UC-commitment protocol between them in order to commit to the input values. After receiving from the token the 
commitments to the input values, David considers the delivery complete, and whenever 
he wants to get his output he just executes an oblivious transfer protocol with the token
with his desired choice bit as input. He checks the correctness of the output using the commitment. 
The crucial point for the simulation to go through is that the simulator should be able
to extract the first commitment before its opening, so that he can choose the common 
reference string as he wishes. In order to accomplish that in face of a potentially 
malicious token which possibly only correctly answers queries to certain values, 
we will commit to a message $m$ by using $\secpar$ pairs of random shares $(v_{i,0},v_{i,1})$ 
where for each pair $v_{i,0}+v_{i,1}=m$. During the committing phase, 
$\david$ interacts with the token and can choose to receive either $v_{i,0}$ 
or $v_{i,1}$ for each pair. To open the commitment, $\goliath$ reveals all the shares.
The specification of the token can be found in \figref{fig-tok3} and of the protocol in \figref{fig-prot2}.

\begin{figure}[tb]
\begin{small}
\begin{tabular}{|p{\textwidth}|}
\hline
\begin{center}
\textbf{Protocol}
\smallskip
\end{center}\\
Parametrized by a security parameter $\secpar$.\\
\\
\textbf{Deliver.} $\goliath$ generates a pair of signing $sk$ and public $pk$ keys for a signature 
scheme. Then he picks a random message $m' \getsr \Ftwo^{\secpar}$ and random vectors
$v_{i,0} \getsr \Ftwo^{\secpar}$ for $i=1, \ldots, \secpar$ and sets $v_{i,1}=m'-v_{i,0}$. He creates the token 
$\token$ (described in \figref{fig-tok3}) with the hardwired vectors $(v_{i,0},v_{i,1})$, $s_0,s_1$ and $pk$, and sends it to $\david$. 
Upon receiving the token \token, $\david$ queries it with random bits $k_i$ for $i=1, \ldots, \secpar$ in order to 
get $v_{i,k_i}$. $\david$ picks a random message $m'' \getsr \Ftwo^{\secpar}$
and sends it to $\goliath$. Then $\goliath$ opens the commitment to $m'$ by sending all the shares 
$(v_{i,0},v_{i,1})$ to $\david$. $\david$ checks if $m'=v_{i,0}+v_{i,1}$ for all $i=1, \ldots, \secpar$, 
aborting the protocol if this is not the case. Both \goliath and \david use $m= m' + m''$ to generate a common reference 
string $\mathsf{crs}$. $\goliath$ signs $\mathsf{crs}$ with his signing key $sk$ and sends the signature $\sigma$ to $\david$.
$\david$ sends $\mathsf{crs}$ and $\sigma$ to $\token$ in order to receive the commitments to 
$s_0$ and $s_1$.\\
\\

\textbf{Choice Phase.} When $\david$ gets his input $c \in \Ftwo$, he sends the message \textsf{output} 
to $\token$ and executes a computationally UC-secure oblivious transfer protocol with the token using 
the common reference string $\mathsf{crs}$ and with input $c$ in order to get the output
$s_c \| d_c$, where $\|$ denotes concatenation. $\david$ checks the correctness of $s_c$ using the commitment that he 
received previously and the opening information $d_c$.\\
\\
\hline
\end{tabular}
\smallskip
\caption{The computationally secure OTM protocol using one token.}\label{fig-prot2}
\end{small}
\end{figure}

\begin{theorem}
In the model where a malicious Goliath is allowed to send messages to the token, the protocol presented 
in \figref{fig-prot2} UC-realizes the functionality $\calF^{\text{OTM-with-Abort}}$ with computational security.
\end{theorem}

\begin{proof}
The correctness of the protocol can be trivially verified. The simulation for the cases that both parties are corrupted or
no parties are corrupted are trivial. We describe below how the simulation proceeds in the other cases.

\paragraph{Corrupted Sender:} If Goliath is corrupted (and thus also the token), the simulator will simulate an interaction 
of the protocol with the adversary and has to extract both $s_0$ and $s_1$ from this interaction in order 
to give them as input for the OTM functionality. The main reason to do this is that the simulator should
be able to extract the value $m'$ before sending $m''$, so that he can choose the common reference string 
$\mathsf{crs}$ as he wishes, thus being able to create a trapdoor to extract $s_0$ and $s_1$ from the committed values.

We have that only Goliath can program the token, so the environment machine will provide the code to Goliath 
(and hence to the simulator). To extract the value $m'$ the simulator does the following. When the commitment step
happens, whenever David sends a valid message \Tchallenge{k_j} to receive a share $\tilde{v}_{j,k_j}$, the simulator first executes
the token with the input $1-k_j$, obtaining an answer $\tilde{v}_{j,1-k_j}$, and then resets the token to the point before this query and 
executes the token with input $k_j$ to obtain $\tilde{v}_{j,k_j}$ and forward it to David. Let $\tilde{m}'_j=\tilde{v}_{j,0}+\tilde{v}_{j,1}$. 
After all the $\secpar$ challenges are done, the simulator fixes $\tilde{m}'$ as the value that appeared more often in the
tuple $(\tilde{m}'_1, \ldots, \tilde{m}'_\secpar)$. He then chooses $m''=m-\tilde{m}'$ for any $m$ he wants.
Lets now analyze this extraction procedure. Let $(\hat{v}_{j,0},\hat{v}_{j,1})$ denote the values that Goliath
reveals in the opening phase. Note that the protocol will be aborted unless $\hat{v}_{j,0}+\hat{v}_{j,1}=\hat{m}'$
for all $j$ and some fixed message $\hat{m}'$. For any $j$, if $\hat{v}_{j,0}\neq \tilde{v}_{j,0}$ and 
$\hat{v}_{j,1}\neq \tilde{v}_{j,1}$ then the protocol will be aborted anyway and we do not need to worry 
about the extracted value. If for the majority of the $j$'s it holds that $\hat{v}_{j,0}= \tilde{v}_{j,0}$ and 
$\hat{v}_{j,1}= \tilde{v}_{j,1}$, then $\tilde{m}'=\hat{m}'$ and thus the extraction procedure works 
properly. The remaining case is the one in which at least half of the $j$'s are such that either $\hat{v}_{j,0}= \tilde{v}_{j,0}$ 
or $\hat{v}_{j,1}= \tilde{v}_{j,1}$, but not both equalities hold. For each such $j$ the probability that the opening check 
succeeds for this pair of vectors is $1/2$ since Goliath cannot get any information from the token. 
Therefore if half or more of the $j$'s are in this condition, the protocol will abort with overwhelming probability 
in the security parameter $\secpar$. 

Given that the extraction worked properly, the simulator can create the common reference string as he wishes and 
so he is able to have a trapdoor to extract the values $s_0$ and $s_1$ from the commitments and give them as input to 
the OTM functionality. To learn the abort behavior, the simulator simulates, at onset and also after 
each incoming message from Goliath to the token, a choice phase execution between David and the token.
The simulator can then use the Switch Off/Switch On commands to adapt $\calF^{\text{OTM-with-Abort}}$'s 
abort behavior properly.
\\

\paragraph{Corrupted Receiver:} If David is corrupted, the simulator gets to know all David's challenges $k_j$ in the first 
commitment. Hence, after seeing $m''$, he can choose any $m'$ he wants (and thus any resulting $m$ and $\mathsf{crs}$) and appropriate shares 
$(\hat{v}_{j,0},\hat{v}_{j,1})$ that are correct from David's point of view. By picking a common reference string together 
with an appropriate trapdoor, the simulator can learn the choice bit $c$ and query it to the functionality $\calF^{\text{OTM-with-Abort}}$ 
to learn $s_c$. Using the equivocability of the UC-commitment the simulator can find an appropriate opening information
$d_c$ and feed $s_c \| d_c$ to David in the OT protocol.
\end{proof}

Note that the above protocol can be trivially extended to the case of sequential OTMs.

\section{The Case of Outgoing Communication}

In the complementary problem, we consider tokens which have a one-way channel that allow them 
to send messages to Goliath, but which cannot receive any information from Goliath. In this 
scenario we would like to implement $\calF^{\text{OTM}}$. Note that in this case 
Goliath cannot control online the abort behavior of the token. We first show an 
impossibility result for unconditionally secure protocols and then present a computationally secure 
protocol using a single token.

\subsection{Impossibility of Information-Theoretically Secure OT(M)}

\begin{lemma} If the tokens can send messages to Goliath, then there is no protocol 
$\Pi$ that realizes OTM, or even oblivious transfer, with information-theoretic security.
\end{lemma}

\begin{proof} (Sketch) The basic idea is that the malicious tokens send their complete view to Goliath after 
each interaction with David. Thus, independently of whether Goliath or some token receive the last 
protocol message, the combined view of Goliath and the tokens is available to a malicious Goliath.
This directly implies that an OT protocol with information-theoretical security is not possible, 
because the whole model collapses to the two-party case in the stand-alone setting. Either the complete 
transcript of the exchanged messages (which is available to a malicious Goliath) 
uniquely determines the choice-bit $c$ of David or a malicious David can obtain both input bits
 $(s_0,s_1)$, and in both cases the oblivious transfer security is broken.
\end{proof}

We remark that the crucial point here is that for oblivious transfer, it does not matter at which time Goliath gets the complete view,
 i.e.\ it does not matter whether some token or Goliath receive the last message. As soon as he learns the choice bit, 
 the protocol is broken. This argumentation, however, does not rule out information-theoretically UC-secure commitments.

\subsection{Unconditionally Secure Commitment with a Single Token}

The idea here is to commit to a message $m$ by using pairs of random shares $(v_{i,0},v_{i,1})$ 
such that for each pair $v_{i,0}+v_{i,1}=m$, the shares are known to both the token and Goliath.
The commitment phase is done by interactions between David and Goliath, where
for each pair David can choose to receive either $v_{i,0}$ or $v_{i,1}$. In order to guarantee the 
binding property, the opening phase is executed between David and the token: David receives 
an opening key from Goliath and forwards it to the token, who checks it and reveals all the shares 
to David. To guarantee that on the one hand David cannot guess the opening key correctly (and thus open 
the commitment whenever he wants), but on the other hand the opening key does not contain enough 
information to allow the token to learn David's choices during the commitment phase (and thus 
successfully open the commitment to any value), we have opening keys that are random
$\secpar$-bit strings and we use $2 \secpar$ pairs of random shares. This 
commitment scheme is secure, but not yet extractable. In order to get extractability, instead of committing to 
the message itself, we first use the $(\secpar, \secpar/2+1)$-Shamir's secret share scheme
to create $\secpar$ shares $(m_1, \ldots, m_{\secpar})$ of the message, then 
commit to each share using the above scheme (in the opening phase a single opening key of 
$\secpar$-bits is given to the token in order to open all the commitments), but we additionally make David ask the token 
to open $\secpar/2$ shares $m_{n_{1}}, \ldots m_{n_{\secpar/2}}$ (without sending the opening key) already in the commitment phase, which do not reveal any information about 
$m$. The specification of the token can be found in \figref{fig-tok4} and of the protocol in \figref{fig-prot3}.

\begin{figure}[tb]
\begin{small}
\begin{tabular}{|p{\textwidth}|}
\hline
\begin{center}
\textbf{Token} $\token$
\smallskip
\end{center}\\
Parametrized by a security parameter $\secpar$. The token is hardwired with the shares 
$(v_{n,i,0},v_{n,i,1})$ for $i=1, \ldots, 2 \secpar, n=1, \ldots, \secpar$ and an opening key $\mathsf{cok} \in \bits^\secpar$. 
It is initialized with state $\Tstate =\Sready$.\\
\\
\textbf{Shares Opening.} Upon receiving a message \Tchallenge{n_1, \ldots, n_{\secpar/2}} from \david 
check if $\Tstate=\Sready$ and $\{n_1, \ldots, n_{\secpar/2}\} \subset \{1, \ldots, \secpar \}$ are the specifications of 
the shares \david wants to be revealed; else abort. Set $\Tstate \gets \SMcom$. 
Send the message \TMshares{(v_{n_j,i,0},v_{n_j,i,1})_{j=1,\ldots, \secpar/2, i=1,\ldots,2\secpar}} to \david.\\
\\

\textbf{Message Opening.} Upon receiving a message \TMrev{\overline{\mathsf{cok}}} from \david 
check if $\Tstate=\SMcom$ and $\overline{\mathsf{cok}}=\mathsf{cok}$; else abort. Set $\Tstate \gets \SMopen$. 
Send the message \TMopen{(v_{n,i,0},v_{n,i,1})_{n=1,\ldots, \secpar, i=1,\ldots,2\secpar}} to \david.\\
\\
\hline
\end{tabular}
\smallskip
\caption{The token for the commitment protocol with outgoing communication.}\label{fig-tok4}
\end{small}
\end{figure}

\begin{figure}[tb]
\begin{small}
\begin{tabular}{|p{\textwidth}|}
\hline
\begin{center}
\textbf{Commitment Protocol}
\smallskip
\end{center}\\
Parametrized by a security parameter $\secpar$.\\
\\
\textbf{Commitment Phase.} $\goliath$ generates an opening key $\mathsf{cok} \getsr \Ftwo^{\secpar}$. Then he 
generates $\secpar$ shares $(m_1, \ldots, m_\secpar)$ of the message $m$ using Shamir's secret sharing scheme.
For each share $m_n$, $\goliath$ picks random vectors $v_{n,i,0} \getsr \Ftwo^{\secpar}$ for $i=1, \ldots, 2\secpar$ and sets $v_{n,i,1}=m_n-v_{n,i,0}$. 
He creates the token $\token$ (described in \figref{fig-tok4}) with the hardwired $\mathsf{cok}$ and vectors $(v_{n,i,0},v_{n, i,1})$
for $n=1,\ldots, \secpar, i=1,\ldots,2\secpar$, and sends it to $\david$. 
Upon receiving the token \token, $\david$ queries \goliath with random bits $k_{n,i}$ for $n=1,\ldots, \secpar, i=1, \ldots, 2\secpar$ 
in order to get $v_{n, i,k_{n,i}}$. Then $\david$ picks a random subset $\{n_1, \ldots, n_{\secpar/2}\} \subset \{1, \ldots, \secpar \}$ 
and asks the token to reveal $(v_{n_j,i,0},v_{n_j,i,1})$ for $j=1,\ldots, \secpar/2, i=1,\ldots,2\secpar$, which he checks against 
the information he received from \goliath; aborting if they do not match.\\
\\

\textbf{Opening Phase.} $\goliath$ sends to $\david$ the message shares $(m_1,\ldots, m_\secpar)$ and also the commitment opening key $\mathsf{cok}$, which \david 
forwards to \token in order to get all the shares $(v_{n, i,0},v_{n, i,1})$. $\david$ checks if $m_n=v_{n,i,0}+v_{n, i,1}$ for all $i$ and $n$, 
aborting the protocol if this is not the case. Then he reconstructs $m$ from the shares; aborting if $m$ is not uniquely 
determined by the shares. 
\\
\hline
\end{tabular}
\smallskip
\caption{The unconditionally secure commitment protocol using one token for the case of outgoing communication.}\label{fig-prot3}
\end{small}
\end{figure}

\begin{theorem}\label{thm:outgoingcom}
In the model where malicious tokens are allowed to send messages to Goliath,
the protocol presented in \figref{fig-prot3} UC-realizes the commitment functionality $\calF^{\text{COM}}$
with unconditional security.
\end{theorem}

The proof is in Appendix \ref{sec:proofoutgoing}.

\subsection{Computationally Secure OTM with a Single Token}

For the case of computational security, it is possible to obtain an OTM protocol which uses only one token. The approach is 
briefly described below.
Using the ideas from the previous section the parties can compute the commitment functionality, which can then be used to establish 
a common reference string between David and the token. The common reference string in turn 
can be used to run computationally UC-secure commitments and OT protocols between the token and David. 
The token commits to the input values using the computationally UC-secure commitment protocol, 
at which point David considers the deliver complete. 
Afterwards, whenever David wants to obtain his output, he engages in a computationally UC-secure 
OT protocol with the token in order to get the desired output and the commitment verification 
information.

\begin{theorem}\label{thm:outgoingotmshort}
In the model where malicious tokens are allowed to send messages to Goliath,
there is a protocol using a single token which UC-realizes the functionality $\calF^{\text{OTM}}$
with computational security.
\end{theorem}

The description of the token and the protocol, as well as the security proof can be found in Appendix \ref{sec:outcompcont}.

\section{Conclusion}

In this work we investigated a weaker isolation model for tamper-proof hardware, namely one-way 
(broadband) communication channels are allowed either for the token creator to the tokens or in the 
opposite direction. In the case that the tokens can receive incoming communication from their creators 
we showed the following: (1) there is an unconditionally secure One-Time Memory (OTM) protocol using two tokens, 
(2) it is impossible to realize OTM with unconditional security from a single token, 
(3) there is an unconditionally secure oblivious transfer protocol using a single token, 
(4) there is a computationally secure OTM protocol using a single token. 
In the case that the tokens can send outgoing communication to their creator 
we showed the following: (1) it is impossible to realize OTM or oblivious transfer 
with unconditional security, (2) there is an unconditionally secure commitment protocol using a single token, 
(3) there is a computationally secure OTM protocol using a single token.

% --- -----------------------------------------------------------------
% --- The Bibliography.
% --- -----------------------------------------------------------------
%
%\addcontentsline{toc}{section}{Bibliography}
%
%% For final copy, edit main.bbl by adding the command
%% \setlength{\parskip}{-1pt}
%% right after the \begin{bibliography}{10} command for a compacter and prettier bib.
%\ifnum\savespace=0
%  \begin{small}
%  \bibliographystyle{plain}
%  \bibliography{./Bibfiles/abbrev0,./Bibfiles/crypto,./Bibfiles/additional}
%  \end{small}
%\else
%  \begin{small}
%  \def\shortbib{1}
%  \bibliographystyle{abbrv}
%  \bibliography{./Bibfiles/abbrev0,./Bibfiles/crypto,./Bibfiles/additional}

\begin{thebibliography}{10}

\bibitem{TCC:AAGPR14}
Shashank Agrawal, Prabhanjan Ananth, Vipul Goyal, Manoj Prabhakaran, and Alon
  Rosen.
\newblock Lower bounds in the hardware token model.
\newblock In Yehuda Lindell, editor, {\em TCC~2014: 11th Theory of Cryptography
  Conference}, volume 8349 of {\em Lecture Notes in Computer Science}, pages
  663--687, San Diego, CA, USA, February~24--26, 2014. Springer, Berlin,
  Germany.

\bibitem{AC:BCGHKR11}
Nir Bitansky, Ran Canetti, Shafi Goldwasser, Shai Halevi, Yael~Tauman Kalai,
  and Guy~N. Rothblum.
\newblock Program obfuscation with leaky hardware.
\newblock In Dong~Hoon Lee and Xiaoyun Wang, editors, {\em Advances in
  Cryptology -- {ASIACRYPT}~2011}, volume 7073 of {\em Lecture Notes in
  Computer Science}, pages 722--739, Seoul, South Korea, December~4--8, 2011.
  Springer, Berlin, Germany.

\bibitem{C:Brands93}
Stefan Brands.
\newblock Untraceable off-line cash in wallets with observers (extended
  abstract).
\newblock In Douglas~R. Stinson, editor, {\em Advances in Cryptology --
  {CRYPTO}'93}, volume 773 of {\em Lecture Notes in Computer Science}, pages
  302--318, Santa Barbara, CA, USA, August~22--26, 1993. Springer, Berlin,
  Germany.

\bibitem{C:BFSK11}
Christina Brzuska, Marc Fischlin, Heike Schr{\"o}der, and Stefan Katzenbeisser.
\newblock Physically uncloneable functions in the universal composition
  framework.
\newblock In Phillip Rogaway, editor, {\em Advances in Cryptology --
  {CRYPTO}~2011}, volume 6841 of {\em Lecture Notes in Computer Science}, pages
  51--70, Santa Barbara, CA, USA, August~14--18, 2011. Springer, Berlin,
  Germany.

\bibitem{FOCS:Canetti01}
Ran Canetti.
\newblock Universally composable security: A new paradigm for cryptographic
  protocols.
\newblock In {\em 42nd Annual Symposium on Foundations of Computer Science},
  pages 136--145, Las Vegas, Nevada, USA, October~14--17, 2001. {IEEE} Computer
  Society Press.

\bibitem{EC:ChaGoySah08}
Nishanth Chandran, Vipul Goyal, and Amit Sahai.
\newblock New constructions for {UC} secure computation using tamper-proof
  hardware.
\newblock In Nigel~P. Smart, editor, {\em Advances in Cryptology --
  {EUROCRYPT}~2008}, volume 4965 of {\em Lecture Notes in Computer Science},
  pages 545--562, Istanbul, Turkey, April~13--17, 2008. Springer, Berlin,
  Germany.

\bibitem{C:ChaPed92}
David Chaum and Torben~P. Pedersen.
\newblock Wallet databases with observers.
\newblock In Ernest~F. Brickell, editor, {\em Advances in Cryptology --
  {CRYPTO}'92}, volume 740 of {\em Lecture Notes in Computer Science}, pages
  89--105, Santa Barbara, CA, USA, August~16--20, 1992. Springer, Berlin,
  Germany.

\bibitem{TCC:CKSYZ14}
Seung~Geol Choi, Jonathan Katz, Dominique Schr{\"o}der, Arkady Yerukhimovich,
  and Hong-Sheng Zhou.
\newblock (efficient) universally composable oblivious transfer using a minimal
  number of stateless tokens.
\newblock In Yehuda Lindell, editor, {\em TCC~2014: 11th Theory of Cryptography
  Conference}, volume 8349 of {\em Lecture Notes in Computer Science}, pages
  638--662, San Diego, CA, USA, February~24--26, 2014. Springer, Berlin,
  Germany.

\bibitem{EC:CraPed93}
Ronald Cramer and Torben~P. Pedersen.
\newblock Improved privacy in wallets with observers (extended abstract).
\newblock In Tor Helleseth, editor, {\em Advances in Cryptology --
  {EUROCRYPT}'93}, volume 765 of {\em Lecture Notes in Computer Science}, pages
  329--343, Lofthus, Norway, May~23--27, 1993. Springer, Berlin, Germany.

\bibitem{C:DFKLS14}
Dana {Dachman-Soled}, Nils Fleischhacker, Jonathan Katz, Anna Lysyanskaya, and
  Dominique Schr{\"o}der.
\newblock Feasibility and infeasibility of secure computation with malicious
  {PUFs}.
\newblock In Juan~A. Garay and Rosario Gennaro, editors, {\em Advances in
  Cryptology -- {CRYPTO}~2014, Part II}, volume 8617 of {\em Lecture Notes in
  Computer Science}, pages 405--420, Santa Barbara, CA, USA, August~17--21,
  2014. Springer, Berlin, Germany.

\bibitem{TCC:DamNieWic09}
Ivan Damg{\aa}rd, Jesper~Buus Nielsen, and Daniel Wichs.
\newblock Universally composable multiparty computation with partially isolated
  parties.
\newblock In Omer Reingold, editor, {\em TCC~2009: 6th Theory of Cryptography
  Conference}, volume 5444 of {\em Lecture Notes in Computer Science}, pages
  315--331. Springer, Berlin, Germany, March~15--17, 2009.

\bibitem{AC:DamSca13}
Ivan Damg{\r a}rd and Alessandra Scafuro.
\newblock Unconditionally secure and universally composable commitments from
  physical assumptions.
\newblock In Kazue Sako and Palash Sarkar, editors, {\em Advances in Cryptology
  -- {ASIACRYPT}~2013, Part II}, volume 8270 of {\em Lecture Notes in Computer
  Science}, pages 100--119, Bengalore, India, December~1--5, 2013. Springer,
  Berlin, Germany.

\bibitem{TCC:DotKraMul11}
Nico D{\"o}ttling, Daniel Kraschewski, and J{\"o}rn M{\"u}ller-Quade.
\newblock Unconditional and composable security using a single stateful
  tamper-proof hardware token.
\newblock In Yuval Ishai, editor, {\em TCC~2011: 8th Theory of Cryptography
  Conference}, volume 6597 of {\em Lecture Notes in Computer Science}, pages
  164--181, Providence, RI, USA, March~28--30, 2011. Springer, Berlin, Germany.

\bibitem{TCC:DMMN13}
Nico D{\"o}ttling, Thilo Mie, J{\"o}rn M{\"u}ller-Quade, and Tobias Nilges.
\newblock Implementing resettable {UC}-functionalities with untrusted
  tamper-proof hardware-tokens.
\newblock In Amit Sahai, editor, {\em TCC~2013: 10th Theory of Cryptography
  Conference}, volume 7785 of {\em Lecture Notes in Computer Science}, pages
  642--661, Tokyo, Japan, March~3--6, 2013. Springer, Berlin, Germany.

\bibitem{TCC:GLMMR04}
Rosario Gennaro, Anna Lysyanskaya, Tal Malkin, Silvio Micali, and Tal Rabin.
\newblock Algorithmic tamper-proof ({ATP}) security: Theoretical foundations
  for security against hardware tampering.
\newblock In Moni Naor, editor, {\em TCC~2004: 1st Theory of Cryptography
  Conference}, volume 2951 of {\em Lecture Notes in Computer Science}, pages
  258--277, Cambridge, MA, USA, February~19--21, 2004. Springer, Berlin,
  Germany.

\bibitem{C:GolKalRot08}
Shafi Goldwasser, Yael~Tauman Kalai, and Guy~N. Rothblum.
\newblock One-time programs.
\newblock In David Wagner, editor, {\em Advances in Cryptology --
  {CRYPTO}~2008}, volume 5157 of {\em Lecture Notes in Computer Science}, pages
  39--56, Santa Barbara, CA, USA, August~17--21, 2008. Springer, Berlin,
  Germany.

\bibitem{C:GIMS10}
Vipul Goyal, Yuval Ishai, Mohammad Mahmoody, and Amit Sahai.
\newblock Interactive locking, zero-knowledge pcps, and unconditional
  cryptography.
\newblock In Tal Rabin, editor, {\em Advances in Cryptology -- {CRYPTO}~2010},
  volume 6223 of {\em Lecture Notes in Computer Science}, pages 173--190, Santa
  Barbara, CA, USA, August~15--19, 2010. Springer, Berlin, Germany.

\bibitem{TCC:GISVW10}
Vipul Goyal, Yuval Ishai, Amit Sahai, Ramarathnam Venkatesan, and Akshay Wadia.
\newblock Founding cryptography on tamper-proof hardware tokens.
\newblock In Daniele Micciancio, editor, {\em TCC~2010: 7th Theory of
  Cryptography Conference}, volume 5978 of {\em Lecture Notes in Computer
  Science}, pages 308--326, Zurich, Switzerland, February~9--11, 2010.
  Springer, Berlin, Germany.

\bibitem{C:IshSahWag03}
Yuval Ishai, Amit Sahai, and David Wagner.
\newblock Private circuits: Securing hardware against probing attacks.
\newblock In Dan Boneh, editor, {\em Advances in Cryptology -- {CRYPTO}~2003},
  volume 2729 of {\em Lecture Notes in Computer Science}, pages 463--481, Santa
  Barbara, CA, USA, August~17--21, 2003. Springer, Berlin, Germany.

\bibitem{EC:Katz07}
Jonathan Katz.
\newblock Universally composable multi-party computation using tamper-proof
  hardware.
\newblock In Moni Naor, editor, {\em Advances in Cryptology --
  {EUROCRYPT}~2007}, volume 4515 of {\em Lecture Notes in Computer Science},
  pages 115--128, Barcelona, Spain, May~20--24, 2007. Springer, Berlin,
  Germany.

\bibitem{TCC:Kolesnikov10}
Vladimir Kolesnikov.
\newblock Truly efficient string oblivious transfer using resettable
  tamper-proof tokens.
\newblock In Daniele Micciancio, editor, {\em TCC~2010: 7th Theory of
  Cryptography Conference}, volume 5978 of {\em Lecture Notes in Computer
  Science}, pages 327--342, Zurich, Switzerland, February~9--11, 2010.
  Springer, Berlin, Germany.

\bibitem{EC:MovSeg08}
Tal Moran and Gil Segev.
\newblock {David} and {Goliath} commitments: {UC} computation for asymmetric
  parties using tamper-proof hardware.
\newblock In Nigel~P. Smart, editor, {\em Advances in Cryptology --
  {EUROCRYPT}~2008}, volume 4965 of {\em Lecture Notes in Computer Science},
  pages 527--544, Istanbul, Turkey, April~13--17, 2008. Springer, Berlin,
  Germany.

\bibitem{EC:OSVW13}
Rafail Ostrovsky, Alessandra Scafuro, Ivan Visconti, and Akshay Wadia.
\newblock Universally composable secure computation with (malicious) physically
  uncloneable functions.
\newblock In Thomas Johansson and Phong~Q. Nguyen, editors, {\em Advances in
  Cryptology -- {EUROCRYPT}~2013}, volume 7881 of {\em Lecture Notes in
  Computer Science}, pages 702--718, Athens, Greece, May~26--30, 2013.
  Springer, Berlin, Germany.

\bibitem{pap01}
Ravikanth~Srinivasa Pappu.
\newblock {\em Physical One-Way Functions}.
\newblock PhD thesis, MIT, 2001.

\bibitem{C:PeiVaiWat08}
Chris Peikert, Vinod Vaikuntanathan, and Brent Waters.
\newblock A framework for efficient and composable oblivious transfer.
\newblock In David Wagner, editor, {\em Advances in Cryptology --
  {CRYPTO}~2008}, volume 5157 of {\em Lecture Notes in Computer Science}, pages
  554--571, Santa Barbara, CA, USA, August~17--21, 2008. Springer, Berlin,
  Germany.

\bibitem{ICALP:PraSahWad14}
Manoj Prabhakaran, Amit Sahai, and Akshay Wadia.
\newblock Secure computation using leaky tokens.
\newblock In Javier Esparza, Pierre Fraigniaud, Thore Husfeldt, and Elias
  Koutsoupias, editors, {\em ICALP 2014: 41st International Colloquium on
  Automata, Languages and Programming, Part I}, volume 8572 of {\em Lecture
  Notes in Computer Science}, pages 907--918, Copenhagen, Denmark, July~8--11,
  2014. Springer, Berlin, Germany.

\bibitem{rue10}
Ulrich R{\"u}hrmair.
\newblock Oblivious transfer based on physical unclonable functions.
\newblock In {\em TRUST}, pages 430--440, 2010.

\end{thebibliography}
%  \end{small}
%\fi

\begin{small}
\bibliographystyle{plain} 

\end{small}

% --- -----------------------------------------------------------------
% --- The Appendix.
% --- -----------------------------------------------------------------
\appendix

\section{Proof of \thref{thm:outgoingcom}}\label{sec:proofoutgoing}

\begin{proof}
The correctness of the protocol is trivial, as well as the simulation for the cases that both parties are corrupted or
no parties are corrupted. We describe below how the simulation proceeds in the other cases.

\paragraph{Corrupted Sender:} 
We have that only Goliath can program the token, so the environment machine will provide the code to Goliath 
(and hence to the simulator). To extract the value $m$ the simulator does the following when David asks 
the token to open $\secpar/2$ shares $m_{n_{1}}, \ldots m_{n_{\secpar/2}}$ of $m$ during the commitment step. Let $Q_0$ denote the subset that David chose to be opened.
First note that given the token's answer to any opening query and the commitment information received during the interaction with Goliath, 
it is possible to distinguish with overwhelming probability if the token opened the shares correctly or not. So the simulator first tests 
if the query $Q_0$ was answered correctly by the token. If $Q_0$ is not answered correctly by the token, we do not need to worry about the extraction since 
the protocol will be aborted anyway. If it was answered correctly, the simulator also runs internally other executions 
of this procedure for opening half of the shares, in each of them asking a random subset $Q_1, Q_2, \ldots$ (with $|Q_i|=\secpar/2$) of the shares to be opened 
(the token is started in each execution from the same state that it was in just before $Q_0$ and there is a fixed exponential 
upper bound on the number of executions that can be performed). The procedure is repeated until some $Q_j$ is answered 
correctly. Note that the probability that the token answers any of the queries $Q_0, Q_1, \ldots$ correctly is 
the same since they are chosen from the same distribution and let $p$ denote this probability. If $p$ is not negligible, then 
the expected numbers of iterations needed to find a second query $Q_j$ that is answered correctly is polynomial.
From $Q_0$ and $Q_j$ the simulator can recover $\hat{m}$ (the unique value that can possibly be accepted with 
non-negligible probability in the tests performed in the commitment's opening phase). 
If $m \neq \hat{m}$, we do not need to worry about the extraction since 
the protocol will be aborted anyway. Therefore the simulator can 
perform the extraction correctly with overwhelming probability. 
\\

\paragraph{Corrupted Receiver:} If David is corrupted, the simulator gets to know all David's challenges $k_{n,i}$ as well as which shares of 
$m$ were opened. Hence after the commitment phase he can still choose any $\hat{m}$ he wants 
and appropriate shares $(\hat{m}_1, \ldots, \hat{m}_\secpar)$ and $(\hat{v}_{n,i,0},\hat{v}_{n,i,1})$ that will be accepted by
David in the opening phase.
\end{proof}

\section{Computationally Secure OTM with a Single Token}\label{sec:outcompcont}

Here we describe the computationally secure OTM protocol with a single token in the case of outgoing communication. 
It uses the commitment protocol as a building block. We describe the full protocol for the sake of completeness.

As in the protocol presented in \secref{sec:comp}, the idea is to use initial interactions (i.e., during the delivering phase) between Goliath 
and David, and also between David and the token in order to compute the commitment functionality and use it 
to establish a common reference string between David and the token, which can then be use to 
run computationally UC-secure commitment and OT protocols between the token and David.
The token then commits to the input values, at which point David considers the deliver complete. 
Afterwards, whenever David wants to obtain his output, he engages in a computationally UC-secure 
OT protocol with the token in order to get the desired output and the commitment verification 
information. The specification of the token can be found in \figref{fig-tok5} and of the protocol in \figref{fig-prot4}.

\begin{sfigure}
\begin{tabular}{|p{\textwidth}|}
\hline
\begin{center}
\textbf{Token} $\token$
\smallskip
\end{center}\\
Parametrized by a security parameter $\secpar$. The token is hardwired with the shares 
$(v_{n,i,0},v_{n,i,1})$ for $i=1, \ldots, 2 \secpar, n=1, \ldots, \secpar$, the inputs $s_0,s_1$, Goliath's public key
$\mathsf{pk}$ and an opening key $\mathsf{cok} \in \bits^\secpar$. It is initialized with state $\Tstate =\Sready$.\\
\\
\textbf{Shares Opening.} Upon receiving a message \Tchallenge{n_1, \ldots, n_{\secpar/2}} from \david 
check if $\Tstate=\Sready$ and $\{n_1, \ldots, n_{\secpar/2}\} \subset \{1, \ldots, \secpar \}$ are the specifications of 
the shares \david wants to be revealed; else abort. Set $\Tstate \gets \SMcom$. 
Send the message \TMshares{(v_{n_j,i,0},v_{n_j,i,1})_{j=1,\ldots, \secpar/2, i=1,\ldots,2\secpar}} to \david.\\
\\

\textbf{Message Opening.} Upon receiving a message \TMrev{\overline{\mathsf{cok}}} from \david 
check if $\Tstate=\SMcom$ and $\overline{\mathsf{cok}}=\mathsf{cok}$; else abort. Set $\Tstate \gets \SMopen$. 
Send the message \TMopen{(v_{n,i,0},v_{n,i,1})_{n=1,\ldots, \secpar, i=1,\ldots,2\secpar}} to \david.\\
\\

\textbf{Inputs Commitment.} Upon receiving a message \Tcommit{\mathsf{crs},\sigma} from \david 
check if $\Tstate=\SMopen$ and if $\sigma$ is a valid signature of $\goliath$ on $\mathsf{crs}$; else abort. 
Set $\Tstate \gets \SIcom$. Commit to the values $s_0,s_1$ using a computationally 
UC-secure commitment protocol that uses the common reference string $\mathsf{crs}$ and send the 
commitments to \david. Let $d_0,d_1$ denote the information to open the commitments.\\
\\

\textbf{Output.} Upon receiving the message \textsf{output} from \david 
check if $\Tstate=\SIcom$; else abort. Next, set $\Tstate \gets \Sdead$ and 
execute with \david a computationally UC-secure oblivious transfer protocol using 
the common reference string $\mathsf{crs}$ and with inputs $(s_0 \| d_0,s_1\| d_1)$.\\ 
\\
\hline
\end{tabular}
\smallskip
\caption{The token for the OTM protocol with outgoing communication.}\label{fig-tok5}
\end{sfigure}

\begin{sfigure}
\begin{tabular}{|p{\textwidth}|}
\hline
\begin{center}
\textbf{Protocol}
\smallskip
\end{center}\\
Parametrized by a security parameter $\secpar$.\\
\\
\textbf{Deliver.} $\goliath$ generates a pair of signing $\mathsf{sk}$ and public $\mathsf{pk}$ keys for a signature 
scheme and also an opening key $\mathsf{cok} \getsr \Ftwo^{\secpar}$. Then he picks a random message $m' \getsr \Ftwo^{\secpar}$ 
and generates the $\secpar$ shares of it $(m'_1, \ldots, m'_\secpar)$ using Shamir's secret sharing scheme.
For each share $m'_n$, $\goliath$ picks random vectors $v_{n,i,0} \getsr \Ftwo^{\secpar}$ for $i=1, \ldots, 2\secpar$ and sets $v_{n,i,1}=m'_n-v_{n,i,0}$. 
He creates the token $\token$ (described in \figref{fig-tok4}) with the hardwired $s_0,s_1$, $\mathsf{pk}$, $\mathsf{cok}$ and vectors $(v_{n,i,0},v_{n, i,1})$
for $n=1,\ldots, \secpar, i=1,\ldots,2\secpar$, and sends it to $\david$. 
Upon receiving the token \token, $\david$ queries \goliath with random bits $k_{n,i}$ for $n=1,\ldots, \secpar, i=1, \ldots, 2\secpar$ 
in order to get $v_{n, i,k_{n,i}}$. Then $\david$ picks a random subset $\{n_1, \ldots, n_{\secpar/2}\} \subset \{1, \ldots, \secpar \}$ 
and ask the tokens to reveal $(v_{n_j,i,0},v_{n_j,i,1})$ for $j=1,\ldots, \secpar/2, i=1,\ldots,2\secpar$, which he checks against 
the information he received from \goliath; aborting if they do not match. $\david$ picks a random message $m'' \getsr \Ftwo^{\secpar}$
and sends it to $\goliath$. $\goliath$ sends to $\david$ the message shares $(m'_1,\ldots, m'_\secpar)$ and also the commitment opening key $\mathsf{cok}$, which \david 
forwards to \token in order to get all the shares $(v_{n, i,0},v_{n, i,1})$. $\david$ checks if $m'_n=v_{n,i,0}+v_{n, i,1}$ for all $i$ and $n$, 
aborting the protocol if this is not the case. Then he reconstruct $m'$ from the shares; aborting if $m'$ is not uniquely 
determined by the shares. Both \goliath and \david use $m= m' + m''$ to generated a common reference 
string $\mathsf{crs}$. $\goliath$ signs $\mathsf{crs}$ with his signing key $\mathsf{sk}$ and sends the signature $\sigma$ to $\david$.
$\david$ sends $\mathsf{crs}$ and $\sigma$ to $\token$ in order to receive the commitments to 
$s_0$ and $s_1$.\\
\\

\textbf{Choice Phase.} When $\david$ gets his input $c \in \Ftwo$, he sends the message \textsf{output} 
to $\token$ and execute with the token a computationally UC-secure oblivious transfer protocol using 
the common reference string $\mathsf{crs}$ and with input $c$ in order to get the output
$s_c \| d_c$. $\david$ checks the correctness of $s_c$ using the commitment that he 
received previously and the opening information $d_c$.\\
\\
\hline
\end{tabular}
\smallskip
\caption{The computationally secure OTM protocol using one token for the case of outgoing communication.}\label{fig-prot4}
\end{sfigure}

\begin{theorem}\label{thm:outgoingotm}
The protocol presented in \figref{fig-prot4} UC-realizes the functionality $\calF^{\text{OTM}}$
with computational security.
\end{theorem}

\begin{proof}
The correctness of the protocol is trivial, as well as the simulation for the cases that both parties are corrupted or
no parties are corrupted. We describe below how the simulation proceeds in the other cases.

\paragraph{Corrupted Sender:} If Goliath is corrupted (and thus also the token), the simulator should be able 
to extract both $s_0$ and $s_1$ to give them as input to the OTM functionality. In short, the simulator should
be able to extract the value $m'$ before sending $m''$, so that he can choose the common reference string 
$\mathsf{crs}$ as he wishes, thus being able to create a trapdoor to extract $s_0$ and $s_1$ from the committed values. 
The extraction happens in the same way as in the proof in Appendix \ref{sec:proofoutgoing} and the 
expected polynomial time simulator can perform the extraction correctly with overwhelming probability. 

If the extraction works properly, the simulator is able to choose the common reference string and therefore is able to 
have a trapdoor that allows him to extract the values $s_0$ and $s_1$ from the commitments. To learn about the abort behavior 
during the choice phase, the simulator simulates a choice phase between David and the token with random inputs. Hence he is able to forward the correct 
inputs to the OTM functionality. \\

\paragraph{Corrupted Receiver:} If David is corrupted, the simulator gets to know all David's challenges $k_{n,i}$ in the first 
commitment as well as which shares of $m'$ were opened. Hence, after seeing $m''$, he can choose any $\hat{m}'$ he wants (and thus any resulting $m$ and $\mathsf{crs}$) 
and appropriate shares $(\hat{m}'_1, \ldots, \hat{m}'_\secpar)$ and $(\hat{v}_{n,i,0},\hat{v}_{n,i,1})$ that are correct from David's point of view. By picking a common reference string together 
with an appropriate trapdoor, the simulator can learn the choice bit $c$ and query it to the functionality $\calF^{\text{OTM}}$ 
to learn $s_c$. Using the equivocability of the UC-commitment the simulator can find an appropriate opening information
$d_c$ and feed $s_c \| d_c$ to David in the OT protocol.
\end{proof}

\end{document}